\documentclass{amsart}

\usepackage{mymacros}
\usepackage{color}
\usepackage[mathscr]{eucal}

\usepackage{enumerate}
\usepackage{xspace}

\usepackage[applemac]{inputenc}
\usepackage{float}
\usepackage{graphicx}

\usepackage{amsthm}

 \newtheorem{theorem}{\normalfont\scshape Theorem }
 \newtheorem{proposition}{\normalfont\scshape Proposition}

\makeatletter
\usepackage[curve,arrow,matrix,graph,frame,line]{xy}
\makeatother

\usepackage{textcomp}

\newcommand{\TnT}{{\mathbb{T}^n_T}}
\newcommand{\PR}{\mathbb{P}_R}
\newcommand{\QR}{\mathbb{Q}_R}

\newcommand{\grad}{\nabla}

\begin{document}


\title{Finite mechanical proxies for a class of reducible continuum systems 
}

\author{F Cardin and A Lovison }
\address{Dipartimento di Matematica \\
              Universit\`a degli Studi di Padova\\
              Via Trieste, 63 - 35121 Padova (Italy) \\
              Tel.: +39-049-8271438\\
              Fax: +39-049-8271499}
\email{cardin@math.unipd.it}
\email{lovison@math.unipd.it} 

\maketitle

\begin{abstract}

We present the exact finite reduction of a class of   
nonlinearly perturbed wave equations,  based on the  Amann--Conley--Zehnder paradigm.
By solving an inverse eigenvalue problem, 
we establish an equivalence between the spectral finite description derived from A--C--Z and a discrete mechanical  model, a well definite finite spring--mass system. By doing so, we decrypt the abstract information encoded in the finite reduction and obtain a physically sound proxy for the continuous problem. 
\\
\textsc{Subject:}
Mathematical Physics
\end{abstract}


\section{Introduction}

Since the dawn of analytical mechanics, the behavior of continuum materials has been modeled by 
observing a large collections of small particles, in the limit when the number of elements approaches infinity. The archetipal example, the 
equations for the large vibrations of the one dimensional elastic string in the plane,  was first derived by  
Euler in 1744  by considering a chain of $N$ beads connected by springs of length $\ell=L/(N-1)$ and by letting $N\to\infty$ keeping both total mass and length $L$ constant.\footnote{See \cite[Chapter 2.3, page 25]{Antman:1980sh}.}   
Discrete systems converging to continua occur also for other physical problems.
The Sine-Gordon equation is obtained as a thermodynamic limit of a finite discrete system, the Frenkel--Kontorova model, as the number of discrete components goes to infinity \cite{Braun:2004gs,Frenkel:1939li}.
Zabusky and Kruskal  \cite{Zabusky:1965iq} studied the continuum limit of the Fermi-Pasta-Ulam model and obtained the Korteweg-de Vries equation, previously derived \cite{Korteweg:1895um} for describing weakly nonlinear shallow water waves.

Also the applications to continuous problems in scientific calculus are based on an analogous limit of a system of finite elements towards an infinite system.
	The mesh size is chosen according to the desired approximation accuracy, or more often to the available computational resources. 
	
	Of course, there are plenty of successful physical and engineering applications, 
however, one cannot conclude that every solution of a continuum physical problem can be approximated consistently with the solutions of a suitable discretized problem.
	Indeed, as observed in \cite{Antman:1980sh}, even though one can deduce the one dimensional wave equation as the limit of a sequence of discrete ODEs, each one consisting in a chain of beads connected with springs, where the total mass and the length are the same,
	it does not follow that the solutions of the former converge to solutions of the latter in any physically reasonable sense.  
	
	Originally, in a view firstly proposed by Von Neumann \cite{vonNeumann:1944wa}, it seemed that the solutions for the positions of the beads in the ODEs, together with their time derivatives and suitable difference quotients would converge respectively to the position, velocity and strain fields for the PDE. 
	In fact, this convergence is valid where classical smooth solutions for the PDE exist.
	
	On the other hand, where velocity and strain suffer jump discontinuities, acceleration waves or shocks, the solutions of the discrete problem develop high--frequency oscillations that persist in the limit as the number of particles $N$ tend to infinity. Consequently the limiting stress results incorrect.\footnote{Thorough discussion on this issue can be found in \cite{Greenberg:1989qy,Greenberg:1994uq}.}
	
%

	In this paper we actually do not come across this difficulty, in a sense, it is overcome in an inverse direction: we will consider a class of nonlinear wave equations for which we will define a totally equivalent discrete nonlinear system composed by beads and springs.
	We will avoid possible problems with discontinuities in the PDE solutions because we will not rely on thermodynamical limits.
	Actually,  
	because 
	we consider the continuous problem (the PDE) as a starting point, and then we attain an exactly equivalent ODE system, preventing possible approximation errors, as it will be described below in details.  
	
	The approach and the results proposed here rely on the Amann--Conley--Zehnder
reduction (ACZ in what follows), a global Lyapunov--Schmidt technique  \cite{Amann:1981jb,Amann:1980kr,Amann:1974xi,Amann:1980ri,Amann:1979mq,Conley:1983yi,Conley:1976bf} which transforms infinite dimensional variational principles into equivalent finite dimensional functionals. This method has been employed in conjunction to topological techniques, e.g., Conley index, Morse theory, Lusternik--Schnirelmann category and degree theory, for proving results of existence and multiplicity of solutions for nonlinear differential equations, in particular for semilinear Dirichlet problems, Hamiltonian systems and nonlinear wave equations \cite{Amann:1981jb,Amann:1980kr,Amann:1980ri,Amann:1979mq,Conley:1983yi,Bambusi:2001ta,Berkovits:2003xh,Berti:2004wa,Berti:2003dl,Cardin:2008ix,Cappiello:2003qz,Coron:1979km,Coron:1979jf,Degiovanni:2003vz,Llave:2000vy,Lovison:2005yt,Lucia:2003nu,Mancini:1978qh,Nirenberg:1981gn,Rabinowitz:1978mn,Rybicki:2001we,Viterbo:1990wj,Viterbo:1992xv}.

It has been pointed out \cite{Carlo:2007tf}, that the (finite number of) parameters involved in the above exact reduction scheme can be regarded as a sort of {\it collective variables}, 
which represent a rather known and strengthened  way to describe, in molecular dynamics and in other allied fields, complex phenomena occurring on distinct space-time scales inside systems with infinite degrees of freedom. The different scales correspond to the various ways we may watch to the system, for example by observing either local variables which depend only on a few degrees of freedom, or else collective variables which describe the global behavior as a whole. This fruitful point of view has been developed since some years and it is utilized in many branches; we recall, e.g., some references in  statistical mechanics: \cite{Boldrighini:1987ez,Maragliano:2006mf,Yukhnovskiui:1987ft}.

The nonlinear elastic string, when processed with ACZ, produces a nonlinear discrete system composed by (i) a finite set of uncoupled harmonic oscillators, plus (ii) an overall nonlinear coupling term. The finite variables $\mu_k$ can be clearly interpreted as collective variables, but are of spectral nature and completely abstract. Nevertheless, there exists a privileged coordinate transformation which restores a strong physical meaning to the reduced system. Making use of well established results about inverse eigenvalue problems\footnote{See \cite{Arsie:2010tx,Boley:1987vn,Ebenbauer:2008ta,Gladwell:2004yg,Gladwell:2006zr,Ji:1998ys,Nylen:1997kx,Xu:2007ly} and the references therein.}, the linear core of the reduced system, i.e., the set of uncoupled linear oscillators, is transformed in an elastic chain of beads and springs, having precisely the same spectral representation.   

This theoretical excursion is  summarized in  Figure \ref{fig:reduction_scheme}.
By doing so we think to provide a further motivation for the ACZ reduction   coming from the microphysics of the matter. 


\begin{figure}[htbp]
\scalebox{0.8}{
$$
\xymatrix{
*++[F-:<3mm>]+{ \txt{\textbf{ $N$ -- oscillator chain}\\ discrete \\
$L=\frac{1}{2} \sum_{i=1}^N m\, {\dot x_i}^2 - \frac{1}{2}x^\top K x +F(x)$ }} \ar@/^5mm/[rr]^{\txt{thermodynamic\\ limit $N\to\infty$}}
& & *++[F-:<3mm>]+{\txt{\textbf{$\infty$ -- elastic string}\\ continuous \& exact \\
$\mathcal{L}=	\frac{1}{2} \abs{\deinde{}{t}u}^2 - \frac{1}{2} \abs{\nabla_x u}^2 + G(u)$
 }} \ar@/^5mm/[dd]^{\txt{ACZ\\\textit{reduction}}} \\
& \txt{equivalence} \ar[ru] \ar[dl] & \\
 *++[F-:<3mm>]+[F-:<4mm>]+{\txt{\textbf{$\bar{N}$ -- oscillator chain}\\ discrete \& `exact' \\
$\widetilde L=\frac{1}{2} \sum_{i=1}^{\bar{N}} \widetilde{m}_i {\dot x_i}^2 - \frac{1}{2}x^\top \widetilde{K} x +\widetilde{F}(x)$ }}  
 \ar@{<-->}@/^5mm/[uu]^{\txt{physical\\ analogy}}
 &&  *++[F-:<3mm>]+{\txt{ \textbf{$\bar{N}$ -- spectral particles}\\ \textbf{(abstract)}\\ discrete \& `exact'\\
 $ \bar L=\frac{1}{2} \sum_{\abs{k}\leq R} \tonde{{\dot \mu_k}^2 - \abs{k}^2\mu_k^2} +\mathcal{N
}(\mu)$
  }} \ar@/^5mm/[ll]^{\txt{isospectral\\change of coordinates}} \ar@/_5mm/@<3mm>[uu]^{\txt{ACZ\\ \textit{reconstruction}}}
}
$$}
\caption{\small 
A suitable change of variables gives a physically sound interpretation to the ACZ reduction for a class of nonlinear wave equations.
}
\label{fig:reduction_scheme}
\end{figure}
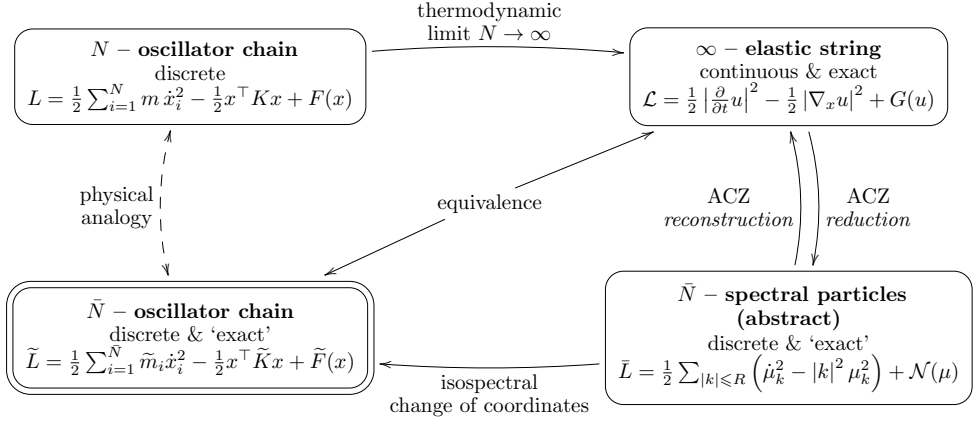


\section{Exact reductions in stationary field theory.  An elastostatics model}
In this section we recall the ACZ reduction for a semilinear elliptic problem, inspired by elastostatics, already described  
in \cite{Cardin:2003rp,Cardin:2007dd,Cardin:2013fk}. The original construction \cite{Amann:1980kr} is reformulated in an Hamiltonian format of 	\cite{Viterbo:1990wj}.
We want to find a
deformation
	$u(x)\in H := H^1_0(\Omega)$, where $x\in\Omega\subseteq \R^n$ is a Stokes
domain (i.e., with a piecewise smooth boundary) satisfying the following
nonlinearly perturbed elliptic equation:
	\begin{equation} \label{eq:DirichletProb}
\begin{cases}
 -\Delta u = F(u), \qquad {\rm in }\ \Omega\\
	u=0 	\qquad\qquad	{\rm on }\ \partial \Omega.
\end{cases}
	\end{equation}
As nonlinear source term $F$ we consider a Nemitski operator
	\begin{eqnarray}
		F:H\To H, \qquad u\mapsto F(u), \\
		F(u)(x) := f\circ u(x),
	\end{eqnarray}
	by a Lipschitz function $f:\R\to\R$, for which
			$$\norm{F(u_1)-F(u_0)}\leq C\norm{u_1-u_0}.$$
	By means of the Amann--Conley--Zehnder reduction, a completely equivalent
\emph{algebraic} equation can be defined:
\begin{equation*}
\begin{cases}
		{\mathcal{E}}({\hat \mu}) = 0, \\
		\hat \mu = (\mu_1,\dots,\mu_\ell) \in \R^\ell.
\end{cases}
\end{equation*}
The ACZ reduction will be now described in detail.
	First of all, a  spectral decomposition of $H^1_0(\Omega)$ w.r.t. $-\Delta$ is
applied,
and the solutions of the Dirichlet problem (\ref{eq:DirichletProb}) can be
represented in the countable dimensional space $\ell^2(\Omega)$,
	$$	H \owns u = u_1 \hat u_1+\dots+u_j \hat u_j + u_{j+1} \hat
u_{j+1}+\dots,	$$
	where the $\hat u_j$ are the eigenvectors of the Laplace operator in $\Omega$:
	$$	-\Delta\hat u_j = \lambda_j \hat u_j, \qquad
0=\lambda_0<\lambda_1\leq\lambda_2\leq\dots,\qquad \forall j=1, 2,\dots.$$
	The spectral decomposition allows to write the inverse operator of the
Laplacian:
	\begin{align*}
	g:H\To H, \qquad \sum u_j \hat u_j \longmapsto \sum \frac{u_j}{\lambda_j} \hat u_j, \\
	  -\Delta \circ g (u) = g\circ (-\Delta) (u) = u.
\end{align*}
	The equation (\ref{eq:DirichletProb}) can be rewritten setting $u= g(v)$,
	$$ v = F(g(v)), \qquad v\in H, \qquad v=\sum_{j=1}^\infty v_j \hat u_j, $$
	and a splitting into a finite dimensional \emph{core} and an infinite
dimensional \emph{tail} can be performed for every cutoff $R\in \N$, applying the
projection operators derived from the spectral decomposition:
	\begin{align}
		v & = v_1\hat u_1 + v_2 \hat u_2 + \dots = \\
		 & =  \PR v + \QR v = \mu + \eta = \\
		 &:= \underbrace{\mu_1\hat u_1+\dots + \mu_R \hat u_R}_{\textstyle{\mu}}+
		\underbrace{\eta_{R+1} \hat u_{R+1} + \dots }_{\textstyle{\eta}} 
	\end{align}
	\begin{equation}
\begin{cases}
			\mu = \PR F(g(\mu+\eta)) ,\qquad {\rm finite \ core}\\
			\eta = \QR F(g (\mu+\eta)),\qquad {\rm infinite\ tail} 
\end{cases}
	\end{equation}
	The key point now is to show that for sufficiently large values of $R$ the
infinite part of the equation is always uniquely solved for every fixed finite
part $\mu$. This is proved if we can show that the operator
	$$ \eta\mapsto \QR F(g (\mu+\eta)) $$
	is contractive for $R$ sufficiently large. Indeed,
	\begin{eqnarray}
	 \norm{\QR F(g (\mu+\eta_1)) - \QR F(g (\mu+\eta_2))} \leq
	 	\norm{F(g (\mu+\eta_1)) -  F(g (\mu+\eta_2))}\leq \\
		\leq C  \norm{g (\mu+\eta_1) - g (\mu+\eta_2)} \leq \frac{C}{\lambda_{R+1}}
		\norm{\eta_1- \eta_2}.
	\end{eqnarray}
	We have $\frac{C}{\lambda_{R+1}} < 1$ for $R$ suitably large because $\lambda_N
\to \infty$.

	Thus we can denote by $\bar \eta (\mu)$ the unique fixed point of the previous
contraction, and we can (at least formally) substitute in the finite equation,
which definitively does represent the very equation of our problem, the determination of $\mu = (\mu_1, \dots, \mu_R)$:
	\begin{equation}\label{eq:ReducedDirichlet}
		\mu = \PR F( g(\mu+\bar \eta (\mu))), 
		(\mu_1,\dots,\mu_R) \in \R^R.
	\end{equation}

	\subsection{Application of the reduction to the variational principle in the
static case}
	By applying the Volterra--Vainberg Theorem \cite{Ambrosetti:1992wm,Vainberg:1964ce,Volterra:1913ce}
		one can associate to the original Dirichlet problem (\ref{eq:DirichletProb}) a
variational principle, an Euler--Lagrange functional:
\begin{equation}
		J[u(\cdot)] := \int_{t=0}^{t=1}
		\contrazione{-\Delta(t u) - F(t u), u} dt=
		 \int_{\Omega}\Big[\frac{1}{2}  |\nabla u|^2-\int^u_{0}  f(s)ds \Big]dx
	\end{equation}
which critical points $dJ[u]=0$ are the (weak) solutions of
(\ref{eq:DirichletProb}).

	The reduction applies as well to the variational principle, simply substituting
$u$ by $g(\mu+\bar\eta(\mu))$, obtaining a finite dimensional functional:
	\begin{equation}
		{\widetilde J}:\R^R\To \R, \qquad {\widetilde J}(\mu) := J[g (\mu+\bar\eta(\mu))].
	\end{equation}
	Indeed one can show this straightforwardly obtained functional is the energy
functional of the finite dimensional equation (\ref{eq:ReducedDirichlet}), in
the sense that the respective solution sets coincide:
	\begin{equation}
		d{\widetilde J}(\mu)=0\ \   \iff \ \ \mu = \PR F( g(\mu+\bar \eta (\mu))).
	\end{equation}
	As already observed in the introduction, this construction lends itself naturally to topological methods, which can be exploited very deeply in finite dimensional spaces 
	for obtaining existence and multiplicity of solutions results.

\section{From statics to dynamics}

In this section we apply the ACZ reduction for a class of nonlinear wave equations on the $n$--dimensional torus. Many results exist in literature for the nonlinear vibrating string 
(\cite{Amann:1981jb,Amann:1980ri,Berkovits:2003xh,Berti:2004wa,Berti:2003dl,Rybicki:2001we,Wang:1989zi} among others), but only few deal with the higher dimensional case \cite{Berti:2010pt,Bourgain:1995iw} using approaches analogous to ACZ. 

\subsection{A Perturbed Wave Equation}\label{w}
	Let us consider the following class of nonlinear wave equations:
we want to find a $T$--periodic motion $u(t,x)\in  H^1_0([0,T]\times\mathbb{T}^n)$
($\ \mathbb{T}^n$ is the $n$--dimensional torus) such that
			\begin{equation}\label{eq:WaveEq}
\begin{cases}
 \square u (t,x) =  F(u), \qquad (t,x) \in [0,T]\times \T^n,\\
				u(0,x) = u(T,x) = 0,	\qquad x \in  \T^n,
\end{cases}
\end{equation}
where $\square:= 	\deindesecuno{}{t} - \Delta_x $ is the (D'Alembertian) wave
operator, while $F$ is a Nemitski operator satisfying the Lipschitz condition
			$$\norm{F(u_1)-F(u_0)} < C\norm{u_1-u_0},$$
as in the static case of Section 1.

	We will denote by $\T^n_T:= [0,T]\times\T^n$ and by $H$ the Hilbert space
\begin{equation}\label{H}
H^1_0(\T^n_T,\R)=\set{u\in\T^n_T\taleche u(0,x)=u(T,x)=0,\forall x\in\T^n}
\end{equation}
endowed with the scalar product
	\begin{align}
	{\contrazione{u,v}}_H & := \contrazione{u,v}_{L^2}
	+ \contrazione{\partial_t u, \partial_t v}_{L^2}
	+ \contrazione{\grad_x u, \grad_x v}_{L^2} \cr
	& =\frac{1}{(2\pi)^n T} \int_{\T^n_T}
		\tonde{uv+\dot u\dot v+ \grad_x u \cdot\grad_x v}
		dt dx^n,
	\end{align}
and the norm
	\begin{equation}
		\norm{u}^2_H := \frac{1}{(2\pi)^n T} \int_{\T^n_T}
		\abs{u}^2 + \abs{\partial_t u}^2	+ \abs{\grad u}^2
		dt d^n x.
	\end{equation}
The environment  $H=H^1_0$ requires a distributional extension of problem
(\ref{eq:WaveEq}).
\begin{proposition}
	The problem
	\begin{equation}\label{eq:WeakWaveEq}
		\frac{1}{(2\pi)^n T} \int_0^T
		\int_{\T^n}\quadre{
		\deinde{}{t}u
		\deinde{}{t}h - \grad_x u \cdot \grad_x h - F(u)h
	}d^n x dt =0, \qquad \forall h\in H
	\end{equation}
	extends to $H$ the wave equation (\ref{eq:WaveEq}).
\end{proposition}

	We will see it is possible to perform on the $x$ spatial coordinates an
analogous \emph{from--infinite--to--finite} dimensional reduction of the wave
equation  (\ref{eq:WaveEq}).

More precisely, we will prove the following theorem.
\begin{theorem}[Reduction of the wave equation]\label{teo:redWave}
	For every Nemitski operator $F:H\to H$ with Lipschitz constant $C>0$ there
exists $\N\owns R=R(C,T,n)>0$ and a family of functions
$\phi_k:\R^\ell\times[0,T]\to\R$ such that the partial differential equation
(\ref{eq:WaveEq}) is equivalent to a nonlinear system of ordinary differential
equations:
	\begin{equation}\label{eq:WaveRed}
\begin{cases}
 \ddot{ \mu}_k(t) + \abs{k}^2\mu_k(t) = \phi_k(\mu,t), \\
	\mu_k(0) = \mu_k(T) = 0,  \\
	\text{ for } k\in\Z^n,\ \abs{k}\leq R.
\end{cases}
\end{equation}
To every (strong) solution $u\in H$ of problem (\ref{eq:WaveEq}), there exists a
(strong) solution $(\mu_1(\cdot),\dots,\mu_\ell(\cdot))$ of (\ref{eq:WaveRed})
and viceversa.
The dimension $\ell$ of the equivalent reduced system is given by
$$\ell:= \#\set{k\in\Z^n :\ \abs{k}\leq R}.$$
\end{theorem}
\begin{proof} [Proof of Theorem \ref{teo:redWave}]
The proof will be split into the following steps, starting from
\begin{equation}\label{eq:startingpoint}
	\deindesecuno{u}{t} - \Delta_x u =  F(u).
\end{equation}
\begin{enumerate}
	\item{	Let $g:H\to H$ be the inverse of the D'Alembert operator 
$\Box u= \deindesecuno{u}{t} - \Delta_x u $, i.e.,  $\square\tonde{ g(v)} = g \tonde{ \square (v)} = v  $. Conjugation phenomena are avoided by choosing $T\not\in \N\pi$. Applying $g$ to both sides of (\ref{eq:startingpoint})	we get $\Box
u = F(u) \quad \Longrightarrow  \quad u = g(F(u))$.}

	\item	{
	By splitting $H:=H^1_0([0,1]\times\T^n)$ into a `finite' and an `infinite'
dimensional subspaces,
		$H = \PR H \oplus \QR H$,
		we get $ \quad \Longrightarrow  \quad   \mu+\eta = g(F(\mu+\eta))$,
		with 	$\mu\in\PR H$ and $\eta\in\QR H$.
			}
	\item	{Applying the Lipschitz property of the Nonlinear perturbation $F:H\to
H$, we get $ \quad \Longrightarrow   \quad  \eta\mapsto \QR g(F(\mu+\eta))$	is a
\emph{contraction} $\forall \mu\in \PR H $ fixed.}
	\item {By substituting the fixed point $\tilde\eta(\mu)$	 into the above
`finite' part
	we get
	 the so-called \emph{bifurcation equation}
	 $ \mu =\PR g(F(\mu+\tilde\eta(\mu)))$
	 and then the reduced system.}
	\end{enumerate}

\end{proof}

\section{Reduction and variational principles}

In this section we will show how the finite dimensional reduction also  extends
to a variational formulation of  (\ref{eq:WaveEq}). Indeed, it is possible to
substitute the formula for $\tilde u(\mu)$ in the variational principle and
exhibit a bijection between the critical points  of the finite and the infinite
dimensional functionals.

Moreover, an accurate analysis of the finite functional will reveal more
explicitly the relation between the nonlinear continuum and the equivalent
finite particles system.

\subsection{Variational Formulation of the Wave Equation}
We will now exhibit the variational principle associated to the wave equation.
We denote by $G$ a primitive of $F$, i.e., $G(s):=\int_0^s F(\tau) d\tau$, the
we write the Euler--Lagrange Functional:
\begin{align}
		J:  H  \To \R, \cr
		\phantom{J: } u \longmapsto J[u]:=\frac{1}{(2\pi)^n T} \int_{ [0,T] \times \T^n}
		\quadre{
			\frac{1}{2} \abs{\deinde{}{t}u}^2 - \frac{1}{2} \abs{\grad_x u}^2 + G(u)
		} dt d^n x.
\end{align}
\begin{theorem}\label{teo:fullVariat}
		Every critical point of $J$ is a (weak) solution of (\ref{eq:WaveEq}) and vice
versa.
		\begin{equation}\label{eq:VariationalEquivalence}
			dJ[u]\cdot h = 0, \qquad \forall h\in H, \qquad \Longleftrightarrow \qquad
\begin{cases}
 \Box u = F(u),\cr
				u(0,x) = u(T,x) = 0 \\
						\textrm{ weakly }
	
\end{cases}
		\end{equation}
		where \emph{weakly} means in the sense of (\ref{eq:WeakWaveEq}).
\end{theorem}

\subsection{Reduction applied to the variational principle}
	We will show that substituting $u\to \tilde u(\mu):=\mu+\tilde\eta(\mu)$ in
$J[u]$, we obtain a \emph{reduced} variational principle, which critical points
will correspond one to one to the (weak) solutions
	of the reduced equation (\ref{eq:WaveRed}).
	\begin{align}\label{eq:redVariatPrin}
		I : H^1_0([0,1];\R^m) \To \R, \\
		\qquad \mu \longmapsto  I[\mu] := J[\tilde u(\mu)]= J[\mu+\tilde\eta(\mu)].
	\end{align}
Theorem \ref{teo:fullVariat} states that the critical points of the finite
functional correspond one to one to the weak solutions of the wave equation
(\ref{eq:WaveEq}).

\begin{proposition}
		Every critical point of $J[u]$ corresponds to a critical point of $I[\mu]$ and
vice versa.
	i.e.,
	\begin{eqnarray}
		dI[\mu]=0 &\Longrightarrow dJ[ u]=0, \qquad  u= \mu+\tilde\eta(\mu), \\
		 dJ[u]=0 &\Longrightarrow dI[\mu]=0 , \qquad \mu=\PR u.
	\end{eqnarray}
\end{proposition}
%
%
%
\section{From continua to discrete. Physical interpretation}
The aim of this section consists in rewriting the reduced energy functional
(\ref{eq:redVariatPrin}) as the Hamilton variational principle corresponding to
a certain discrete Lagrangian system, i.e.,
$$I[\mu] = \frac{1}{T} \int_0^T \mathcal{L}(\mu(t),\dot\mu(t)) dt. $$
 The Lagrangian system associated with $\mathcal{L} (\mu,\dot\mu)$ will be
discussed  and physically interpreted.

 Let us start writing more explicitly the reduced functional $I(\mu)$, setting $\TnT=[0,T]\times \T^n$, 
	\begin{align} \label{eq:redLagrUno}
 I[\mu] := J[\mu+\tilde\eta(\mu)] 
		= \frac{1}{(2\pi)^n T}  \int_\TnT \quadre{
		\frac{1}{2} \abs{\deinde{\mu}{t} }^2 - \frac{1}{2} \abs{\grad_x \mu}^2 +
		G(\mu + \tilde\eta(\mu))
		} d^n x dt + \\
		+ \frac{1}{(2\pi)^n T}  \int_\TnT \quadre{
		\frac{1}{2} \abs{\deinde{\tilde\eta(\mu)}{t}}^2 -  \frac{1}{2}\abs{ \grad_x
\tilde\eta(\mu)}^2
		 } d^n x dt,
	\end{align}
where the mixed terms in the expansion of the squared terms have vanished after
integration, being $\mu$ and $\tilde\eta(\mu)$ orthogonal, as well as their
derivatives.

 Let us now consider a weak formulation of the fixed point:
 	\begin{equation*}
		 \frac{1}{(2\pi)^n T}
		\int_\TnT \quadre{
			\deinde{}{t} \tilde\eta(\mu) \deinde{}{t} h - \grad_x\tilde\eta(\mu) \cdot
\grad_x h + \QR F(\mu+ \tilde\eta(\mu)) \cdot h
		} d^nx dt = 0, \quad \forall h\in\QR H.
	\end{equation*}
Note that we can write $F$ in place of $\QR F$, being $h$ orthogonal to $\PR F$.
If we substitute $\tilde\eta(\mu)$ for $h$, we get:
\begin{equation}
	\int_\TnT \quadre{
			\abs{\deinde{}{t} \tilde\eta(\mu)}^2 \!\!\!
			 - \abs{\grad_x\tilde\eta(\mu)}^2
			} d^nx dt
			= - \int_\TnT\!\!\!\!\! F(\mu+ \tilde\eta(\mu)) \cdot \tilde\eta(\mu)
		 d^nx dt, \quad \forall h\in\QR H.
	\end{equation}
Substituting into the reduced functional (\ref{eq:redLagrUno}),  we get
	\begin{equation} 
		I[\mu] = \frac{1}{(2\pi)^n T}
		\int_\TnT \Bigl[\frac{1}{2} \abs{\deinde{\mu}{t} }^2 - \frac{1}{2}
\abs{\grad_x \mu}^2 +  G(\mu + \tilde\eta(\mu)) 
-\ \frac{1}{2}
F(\mu+\tilde\eta(\mu))\tilde\eta(\mu) \Bigr] d^n x dt.
	\end{equation}
Writing the Fourier expansion of $\mu$ and integrating out on $\T^n$ the first
two terms, we obtain
\begin{align}\label{N}
	I[\mu] & = \frac{1}{T}
		\int_0^T \sum_{\abs{k}\leq R}
		\tonde{ \frac{1}{2} \abs{\dot\mu_k (t) }^2 - \frac{\abs{k}^2}{2} \abs{
\mu_k(t)}^2 } dt + \\
		&+  \frac{1}{T}
		\int_0^T  \underbrace{\frac{1}{(2\pi)^n}
		\quadre{ \int_{\T^n} G(\mu + \tilde\eta(\mu)) - \frac{1}{2}
F(\mu+\tilde\eta(\mu))\tilde\eta(\mu) d^n x }}_{=: N[\mu(t)]} dt.
\end{align}
As a result, the Lagrangian function corresponding to the variational principle
(\ref{eq:redVariatPrin}) can be written as:
	\begin{equation}\label{eq:spectral_variables}
		\mathcal{L}(\mu,\dot\mu) :=
	\sum_{\abs{k}\leq R}
		\underbrace{\tonde{
		\frac{1}{2}\abs{	\dot\mu_k}^2  - \frac{\abs{k}^2}{2}\abs{ \mu_k}^2
		}}
		_{\substack{{\rm system\ of}\  \ell \cr {\rm harmonic\ oscillators}}}
		+ \underbrace{\mathcal{N}[\mu].}_{{\rm nonlinear\ coupling}}
	\end{equation}
Thus, the finite dimensional reduced functional can be interpreted as a classical Action with 
Lagrangian corresponding to a lattice consisting in
 $\ell$ harmonic oscillators with displacement $\mu_k$, unit mass, stiffness
$\abs{k}^2$, and a global nonlinear coupling function $\mathcal{N}(\mu)$, explicitly computable from (\ref{N}).

\section{Reconstruction of mass--spring systems from eigenvalue sequences}

Let us consider for simplicity the one dimensional case:
\begin{equation}\label{eq:one_dim_case}
\begin{cases}
 \deindesecuno{}{t}u - \deindesecuno{}{x} u  = F(u), \\
u(0,x) = u(T,x) = 0, \qquad \forall x\in [0,L],\\
u(t,0) = u(t,L) = 0, \qquad \forall t.
\end{cases}
\end{equation}

The eigensystem of $- \deindesecuno{}{x}$ is given by
\begin{align}\label{eq:eig_el_string}
   0 < \lambda_1 < \lambda_2 < \dots \quad 
    \lambda_{k}  = \tonde{\frac{\pi k}{L}}^2 \quad  \dots  \\
 \widehat u_{k}(x) = \sqrt{\frac{L}{2}}\frac{1}{\pi k}\sin\tonde{\frac{\pi k}{L} x}.
\end{align}

The above ACZ construction leads to the following finite dimensional Lagrangian (renormalized with respect to (\ref{eq:spectral_variables})):
\begin{equation}\label{eq:one_dim_lagran}
		\mathcal{L}(\mu,\dot\mu) :=
	\sum_{k=0}^N
		\tonde{
		\frac{1}{2}\abs{	\dot\mu_k}^2  - \frac{1}{2} \tonde{\frac{\pi k}{L}}^2\abs{ \mu_k}^2
		}
		+ \mathcal{N}(\mu).
\end{equation}		
 This exact finite spectral formulation does correspond to the  description of a large number of --more or less realistic-- finite dynamical systems. The direct reconnaissance of (\ref{eq:one_dim_lagran}) shows a system of $N$ unit mass harmonic oscillators, whose positions are described by the spectral variables $\mu_j$, with elastic constants
$\lambda_k=\tonde{\frac{\pi k}{L}}^2$; the term $\mathcal{N}(\mu)$ represents the nonlinear correction.

We recall that the one dimensional wave equation, i.e., equation (\ref{eq:one_dim_case}) with $F\equiv 0$, originates as the thermodynamic limit of a  finite chain of $N$ small beads connected with springs, sending $N$ to infinity while keeping mass density and elastic tension constant \cite{Lovison:2009tb}.
On the other hand, the quadratic part in (\ref{eq:one_dim_lagran}) does not correspond to a chain of springs and beads, but to a set of uncoupled harmonic oscillators.

\smallskip

\par\noindent
{\normalfont\scshape Question}\  {\it 
 Is it possible to change the variables in (\ref{eq:one_dim_lagran}),  obtaining a genuine linear spring-mass system (plus higher order terms) but maintaining the core of the spectral sequence (\ref{eq:eig_el_string}), i.e., $\lambda_k=\tonde{\frac{\pi k}{L}}^2$, $k=1,\dots,N$? 
}

\smallskip

If so, the ACZ reduction would give back a finite mechanical system with a  relevant physical meaning, because of the analogy with the elements of the thermodynamic limit sequence.

To this purpose, we consider the equation of motion of a chain of $N$ masses $m_1,\dots,m_N$  connected with springs $k_i$ and with fixed extrema. We consider also a possible nonlinear contribution $\mathsf{F}(u)$, where $u=(u_1,\dots,u_N)$ is the vector of the displacements of the masses from their rest position.
\begin{align}
	\qquad M \ddot{u} + K u   = \mathsf{F}(u), \\
	M= \quadre{\begin{array}{cccc} {m_1} & 0 & \dots & \cr 
	0 & {m_2} & 0 & \cr
	\vdots & & \ddots & 0\cr
	 & & 0 & {m_N} 
	\end{array}}, \quad &
	K= \quadre{\begin{array}{cccc} 
		k_1+k_2 & -k_2 & \dots & \cr 
		-k_2 & k_2+k_3 & -k_3 & \cr
	\vdots & & \ddots & -k_N\cr
	 & & -k_N & k_N+k_{N+1} 
	\end{array}}
	\label{eq:mass_spring_matrices}
\end{align}
We are interested in the quadratic part of the Lagrangian of this system: 
\begin{equation}\label{eq:mass_spring_lagr}
	L(u,\dot u) =  \frac{1}{2} \dot u^\top M \dot u - \frac{1}{2} u^\top K u,
\end{equation}
and we want to find a linear coordinate transformation $u=\Upsilon v$ such that (\ref{eq:mass_spring_lagr}) becomes the quadratic part of (\ref{eq:one_dim_lagran}), maintaining the same proper oscillation modes. 

If we were to take a system with equal masses and equal springs, as in the sequence of the thermodynamic limit, we would obtain the following eigenvalues:
\begin{equation}\label{eq:eq_mass_eigen}
\tilde\lambda_{k,N+1} = \frac{4(N+1)^2}{L^2} \sin^2\tonde{\frac{\pi}{2(N+1)}k}, \quad \underset{N\rightarrow \infty}{\To} \quad \tonde{\frac{\pi k}{L}}^2,
\end{equation}
which, as shown in figure \ref{fig:eigen_comparisons}, are close to the desired eigenvalues only in the first part of the sequence. Of course the $k$-th eigenvalue of the $N$-th chain, $\tilde\lambda_{k,N+1}$ will converge to $\lambda_k$ as $N$ tends to infinity, because $\frac{\pi}{2(N+1)}k\ll 1$ and the sinus in (\ref{eq:eq_mass_eigen})
can be substituted with its argument. Nevertheless, the largest eigenvalues will maintain a definite distance from the corresponding eigenvalues of the continuum. More precisely, 
\begin{equation}
	\frac{\tilde\lambda_{N,N}}{\lambda_{N}} \cong \frac{4(N+1)^2}{L^2} \scalebox{2.}{$/$} \frac{\pi^2 N^2}{L^2} \to \frac{4}{\pi^2}  (\cong 0.4052847 ).
	\label{eq:eigen_comp_disc_cont}
\end{equation}
Therefore, for matching exactly the whole sequence of eigenvalues we must consider a system with variable masses and springs. However, the reconstructed system can be symmetric, as it will be specified below.
\begin{figure}[htbp]
\begin{center}
\includegraphics[width=80mm]{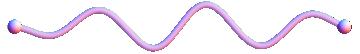}\\ \includegraphics[width=80mm]{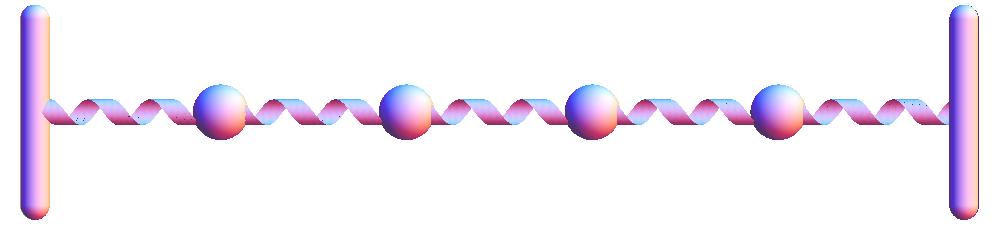} 
\caption{Equivalence among a continuum and a discrete system}
\label{fig:cont-discr-equiv}
\end{center}
\end{figure}

Let us consider the diagonal matrix $D^{-1}:= \mathrm{diag}(\frac{1}{\sqrt{m_1}},\dots,\frac{1}{\sqrt{m_N}})$. Clearly we have that $ M = D^2$, therefore, by setting 
\begin{equation}
u = D^{-1} y,
\end{equation}
we have that (\ref{eq:mass_spring_lagr}) becomes
\begin{equation}
\tilde{L}(y,\dot y) = 	L(u,\dot u) =  \frac{1}{2} \sum_{k=1}^N \dot y_k^2 - \frac{1}{2} y^\top D^{-1} K D^{-1}y.
\end{equation}
If moreover we consider the orthogonal matrix $U$ of the normalized eigenvectors of $D^{-1} K D^{-1}$, by setting 
\begin{equation}
y=Uq,
\end{equation}
the Lagrangian function (\ref{eq:mass_spring_lagr}) is represented by
\begin{equation}
\tilde{\tilde{L}}(q,\dot q) = 	L(u,\dot u) =  \frac{1}{2} \sum_{k=1}^N \dot q_k^2 - \frac{1}{2} q^\top \tilde\Lambda q,
\end{equation}
where the diagonal matrix $\tilde \Lambda := \mathrm{diag}(\tilde\lambda_1,\dots,\tilde\lambda_N)$ contains the eigenvalues of $D^{-1} K D^{-1}$. 

Therefore, the answer to our question is positive if it is possible to choose masses $m_1,\dots,m_N$ and springs $k_1,\dots,k_{N+1}$ such that the proper frequencies  $\tilde\lambda_1,\dots,\tilde\lambda_N$ coincide with the sequence $\lambda_1,\dots,\lambda_N$ in (\ref{eq:eig_el_string}). 

This problem is an \emph{inverse eigenvalue problem}, and has a unique solution in the form of a \emph{persymmetric} system, i.e., a chain of masses and springs symmetric with respect to the mid point\,\footnote{Thorough discussion on the problem of finding special springs--masses systems reproducing specified eigenvalues can be found in \cite{Boley:1987vn,Gladwell:2004yg,Gladwell:2006zr,Ji:1998ys,Nylen:1997kx,Xu:2007ly}}.
\begin{theorem}
 For any sequence of non negative distinct numbers $0\leq \lambda_1 < \lambda_2<\dots<\lambda_N$, there exists a unique persymmetric mass--spring system (\ref{eq:mass_spring_matrices}-\ref{eq:mass_spring_lagr}),  $m_1,\dots,m_N$, $k_1,\dots,k_{N+1}$, with specified total mass $m=\sum_{k=1}^N m_k$, having the sequence $\lambda_1,\dots,\lambda_N$ as the set of proper oscillation modes.
\end{theorem}
\begin{proof}
 The lines of the proof follows \cite[Chapters 3 and 4]{Gladwell:2004yg}. 
 A \emph{Jacobi matrix} is a positive semi-definite symmetric tridiagonal
matrix with (strictly) negative codiagonal entries. 
If the elastic constants $k_1,\dots,k_{N+1}$ are all positive, the stiffness matrix $K$ in (\ref{eq:mass_spring_lagr}) is a Jacobi matrix. Moreover, if also the masses $m_1,\dots,m_N$ are all positive, the matrix $D^{-1}K D^{-1}$ is still a Jacobi matrix. The proof is completed by invoking 
\begin{quotation}
\noindent\cite[Theorem 4.3.2]{Gladwell:2004yg}: \itshape There is a {\rm unique}, up to multiplicative constants, persymmetric Jacobi matrix $J$ for any assigned spectrum $\sigma(J) =
(\lambda_j)_{j=1}^N$, satisfying $0\leq \lambda_1 < \lambda_2< \dots < \lambda_N $.
\end{quotation}
As described in details in \cite[Section 4.4-4.6]{Gladwell:2004yg} it is possible to reconstruct uniquely a spring-mass system with specified total mass from such a persymmetric Jacobi matrix $J$. 
\end{proof}
For instance, if we consider the first five eigenvalues of an elastic string and reconstruct an equivalent persymmetric elastic chain with five beads and total mass 1, we obtain the mechanical system:
\begin{eqnarray}
 \text{\bf masses:} \quad & \left\{\frac{5}{21}, \frac{5}{28}, \frac{1}{6}, \frac{5}{28},  \frac{5}{21} \right\},
 \label{eq:rec_masses}\\
\text{\bf springs:} & \frac{\pi^2}{L^2}\left\{\frac{25}{42},\frac{15}{14},\frac{5}{4},\frac{5}{4},\frac{15}{14},\frac{25}{42}\right\}.\label{eq:rec_springs}
\end{eqnarray}
The corresponding normalized eigenvectors of the discrete \emph{simulacrum} are shown in figure \ref{fig:disc_simulacrum}(b), along with the first five eigenvectors of the continuous system for comparison (c). The eigenvalues of a chain with equal masses and springs are compared with the eigenvalues of the string in figure \ref{fig:eigen_comparisons}, for a system with 5 masses and 25 masses. Being the largest eigenvalue of the finite chain less than half the corresponding eigenvalue of the string (\ref{eq:eigen_comp_disc_cont}), the reconstructed system will not tend to a homogeneous system but a definite difference among masses and springs is maintained as the system size grows.

\begin{figure}[htbp]
\begin{center}
 \includegraphics[width=61mm]{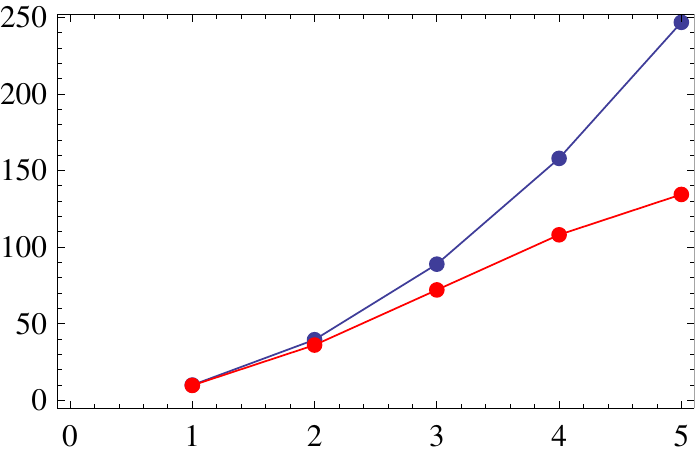} \includegraphics[width=61mm]{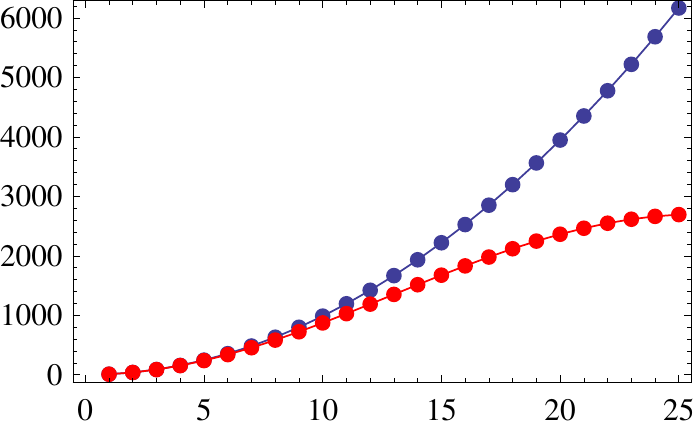} 
\caption{Right panel: first five eigenvalues of the elastic string (blue) and eigenvalues of the elastic chain with five equal masses and springs (red). Left panel: first 25 eigenvalues of the elastic string (blue) and eigenvalues of the elastic chain with 25 equal masses and springs (red).
}
\label{fig:eigen_comparisons}
\end{center}
\end{figure}

\begin{figure}[htbp]
\begin{center}
\parbox[b]{6mm}{(a)\\\rule{0mm}{20mm}}\includegraphics[width=55mm]{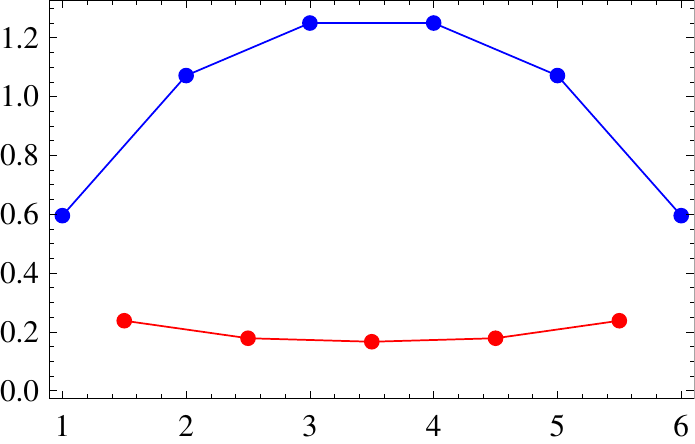} \includegraphics[width=53mm]{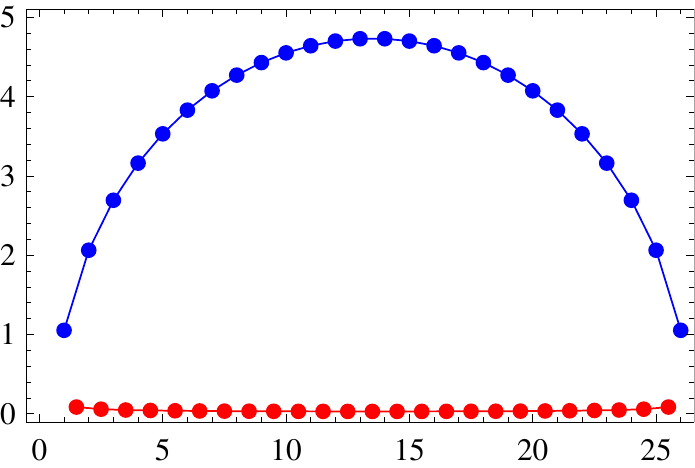} \\
\parbox[b]{16mm}{(b)\\\rule{0mm}{33mm}}\includegraphics[width=110mm]{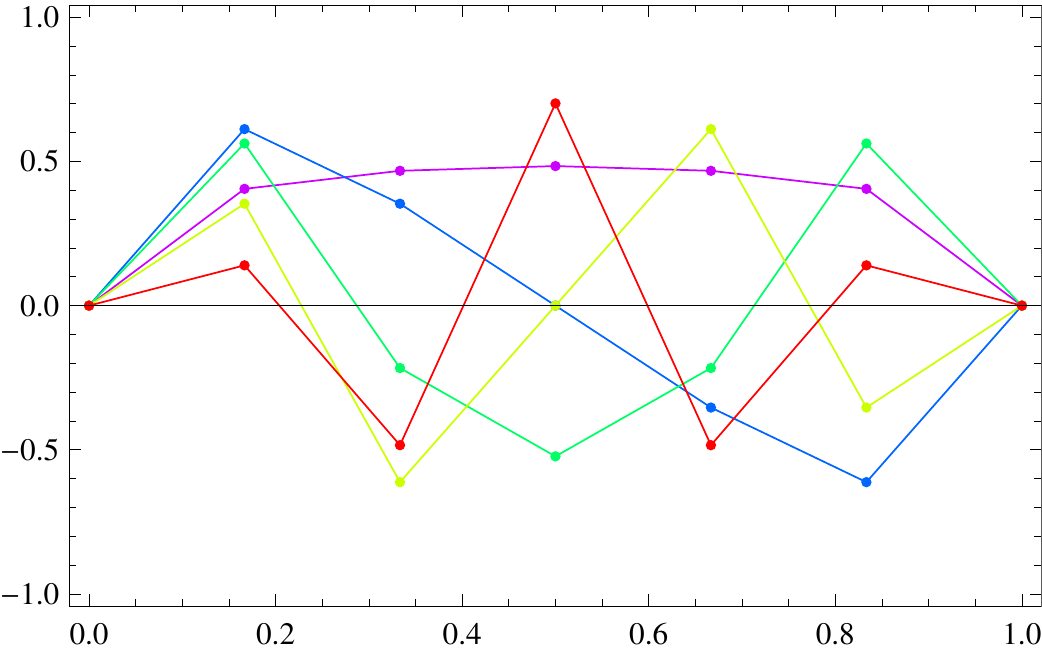}\\
\parbox[b]{16mm}{(c)\\\rule{0mm}{33mm}}\includegraphics[width=110mm]{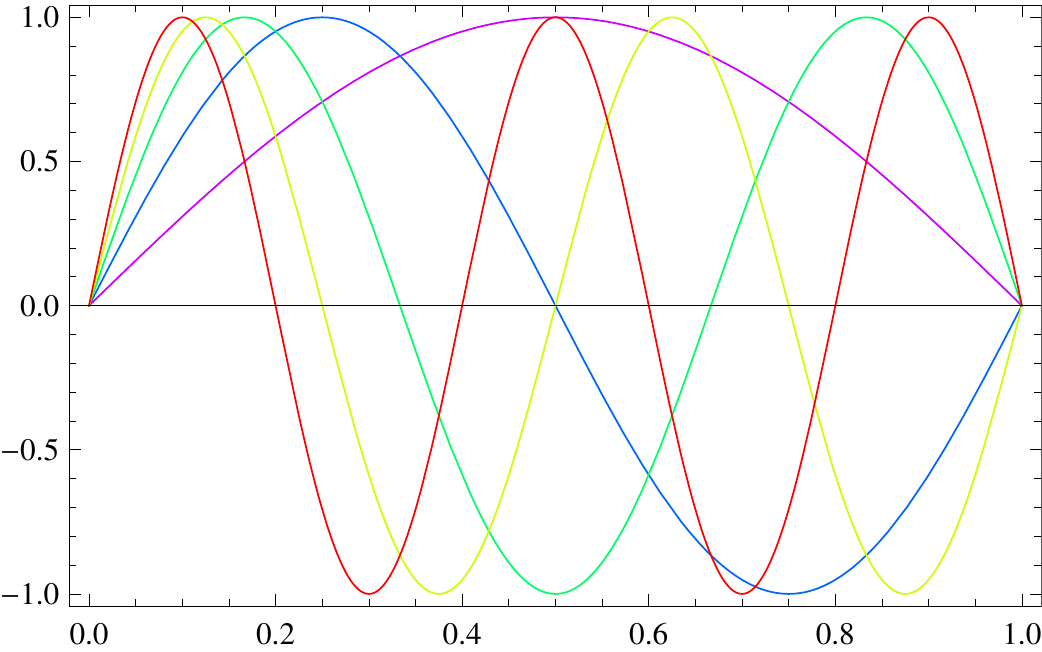}
\caption{Five-beads elastic chain with the eigenvalues equal to the first five eigenvalues of the elastic string. Panel (a): springs (blue) and masses (red) of the reconstructed system (\ref{eq:rec_masses}-\ref{eq:rec_springs}) with five elements (left) and 25 elements (right). 
Panel (b): normalized eigenvectors for the chain with 5 elements. Panel (c): first five eigenvectors for the elastic string. 
}
\label{fig:disc_simulacrum}
\end{center}
\end{figure}

\section{Conclusions}

The above exploration around the meaning of the exact finite reduction of the PDE setting of a static or dynamic continuous system, here sketched by the problems (\ref{eq:DirichletProb}) and (\ref{eq:WaveEq}) respectively, has led us to the following interpretation: such exact reduced system is precisely a discrete  `simulacrum' of the original system, reminiscent of all its main features, and it can be perfectly interpreted as a genuine physical system, strictly analogous to the systems of the hierarchy employed in the thermodynamic limit, unless that the masses and springs are not all equal, but persymmetric, i.e., symmetric with respect to the middle point of the system.


\bibliographystyle{unsrt} 		

\bibliography{Esistenza}

\end{document}